\documentclass[twocolumn,american,aps, pra]{revtex4-1}
\usepackage[T1]{fontenc}
\usepackage[latin9]{inputenc}
\usepackage{babel}

\usepackage{float}
\usepackage{amsthm}
\usepackage{amsmath}
\usepackage{graphicx}
\usepackage[unicode=true, 
 bookmarks=true,bookmarksnumbered=false,bookmarksopen=false,
 breaklinks=false,pdfborder={0 0 0},backref=false,colorlinks=false]
 {hyperref}
\hypersetup{
 pdfauthor={Alexander Streltsov}}
 
\makeatletter

\providecommand{\tabularnewline}{\\}

\@ifundefined{textcolor}{}
{%
 \definecolor{BLACK}{gray}{0}
 \definecolor{WHITE}{gray}{1}
 \definecolor{RED}{rgb}{1,0,0}
 \definecolor{GREEN}{rgb}{0,1,0}
 \definecolor{BLUE}{rgb}{0,0,1}
 \definecolor{CYAN}{cmyk}{1,0,0,0}
 \definecolor{MAGENTA}{cmyk}{0,1,0,0}
 \definecolor{YELLOW}{cmyk}{0,0,1,0}
 }
\theoremstyle{plain}
\newtheorem{thm}{Theorem}
  \theoremstyle{plain}
  \newtheorem{prop}[thm]{Proposition}

\usepackage{lmodern}
\newcommand{\bra}[1]{\langle #1|}
\newcommand{\ket}[1]{|#1\rangle}
\newcommand{\braket}[1]{\langle#1\rangle}
\usepackage[usenames,dvipsnames]{pstricks}
\usepackage{epsfig}

\makeatother

\begin{document}

\title{Easy implementable algorithm for the geometric measure of entanglement}
\begin{abstract}
We present an easy implementable algorithm for approximating the geometric
measure of entanglement from above. The algorithm can be applied to
any multipartite mixed state. It involves only the solution of an
eigenproblem and finding a singular value decomposition, no further
numerical techniques are needed. To provide examples, the algorithm
was applied to the isotropic states of 3 qubits and the 3-qubit XX
model with external magnetic field.
\end{abstract}

\author{Alexander Streltsov}

\email{streltsov@thphy.uni-duesseldorf.de}

\author{Hermann Kampermann}

\author{Dagmar Bruß}

\affiliation{Heinrich-Heine-Universit\"{a}t D\"{u}sseldorf, Institut f\"{u}r
Theoretische Physik III, D-40225 D\"{u}sseldorf, Germany}

\maketitle

\section{Introduction}

Quantum entanglement as a fascinating nonclassical feature has attracted
attention since the early days of the quantum theory \cite{Einstein1935,Schroedinger1935}.
In the last decades its importance for the quantum information theory
has been recognized, since entanglement plays a crucial role in almost
every quantum computational task \cite{Nielsen2000}. 

A bipartite pure state is said to be entangled, if it cannot be written
in the product form \begin{equation}
\ket{\psi_{\mathrm{sep}}^{AB}}=\ket{\psi^{A}}\otimes\ket{\psi^{B}}.\end{equation}
States which are not entangled are called separable. In general, the
number of parties is $n\geq2$, and fully separable pure states become
\begin{equation}
\ket{\psi_{\mathrm{sep}}}=\otimes_{i=1}^{n}\ket{\psi^{\left(i\right)}}.\end{equation}
The theory of entanglement has also been extended to the case, when
the quantum state is not pure \cite{Werner1989,Horodecki2009}. Then
a mixed state $\rho_{\mathrm{sep}}$ is called separable, if it can
be written as a convex combination of separable pure states: \begin{equation}
\rho_{\mathrm{sep}}=\sum_{i}p_{i}\otimes_{i=1}^{n}\ket{\psi^{\left(i\right)}}\bra{\psi^{\left(i\right)}}\end{equation}
with nonnegative probabilities $p_{i}$, $\sum_{i}p_{i}=1$. Quantification
of entanglement is one of the main research areas in quantum information
theory \cite{Horodecki2009}. For bipartite pure states, the entanglement
is usually quantified using the von Neumann entropy of the reduced
state: \begin{equation}
E\left(\ket{\psi^{AB}}\right)=-\mathrm{Tr}\left[\rho^{A}\log_{2}\rho^{A}\right],\label{eq:E}\end{equation}
where $\rho^{A}=\mathrm{Tr}_{B}\left[\ket{\psi^{AB}}\bra{\psi^{AB}}\right]$.
For multipartite systems and mixed states many different measures
of entanglement were proposed \cite{Plenio2007,Horodecki2009}. In
general, a measure of entanglement is any continuous function $E$
on the space of mixed states $\rho$ which satisfies at least the
following properties \cite{Horodecki2009}: 
\begin{itemize}
\item $E$ is nonnegative and zero if and only if the state is separable.
\item $E$ does not increase under local operations and classical communication:
\[
E\left(\Lambda\left(\rho\right)\right)\leq E\left(\rho\right),\]
where $\Lambda$ is any LOCC operation.
\end{itemize}
For bipartite mixed states, an important measure of entanglement is
the entanglement of formation $E_{f}$. For pure states it is defined
as the von Neumann entropy of the reduced state as given in (\ref{eq:E}).
The extension to mixed states is done via the convex roof construction
\cite{Uhlmann1997,Bennett1996}: \begin{equation}
E_{f}\left(\rho\right)=\min\sum_{i}p_{i}E\left(\ket{\psi_{i}}\right),\end{equation}
where the minimum is taken over all pure state decompositions of $\rho$.

In this paper we will consider the geometric measure of entanglement.
For pure states it is defined as follows \cite{Wei2003}: \begin{equation}
E_{G}\left(\ket{\psi}\right)=1-\max_{\ket{\phi}\in S}\left|\braket{\psi|\phi}\right|^{2},\end{equation}
where the maximization is done over the set of separable states $S$.
For mixed states $\rho$ the geometric measure of entanglement was
originally defined via the convex roof construction, in the same way
as it was done for the entanglement of formation \cite{Wei2003}:
\begin{equation}
E_{G}\left(\rho\right)=\min\sum_{i}p_{i}E_{G}\left(\ket{\psi_{i}}\right)\label{eq:Eg}\end{equation}
with minimization over all pure state decompositions of $\rho$. Similar
measures of entanglement were also considered earlier in \cite{Shimony1995,Barnum2001}. 

If $\rho$ is a two-qubit state, general expressions for $E_{f}$
and $E_{G}$ are known \cite{Wootters1998,Vidal2000,Wei2003}: \begin{eqnarray}
E_{f}\left(\rho\right) & = & h\left(\frac{1}{2}+\frac{1}{2}\sqrt{1-C\left(\rho\right)^{2}}\right),\\
E_{G}\left(\rho\right) & = & \frac{1}{2}\left(1-\sqrt{1-C\left(\rho\right)^{2}}\right).\label{eq:twoqubits}\end{eqnarray}
The concurrence $C\left(\rho\right)$ is given by \begin{equation}
C\left(\rho\right)=\max\left\{ 0,\lambda_{1}-\lambda_{2}-\lambda_{3}-\lambda_{4}\right\} ,\end{equation}
where $\lambda_{i}$ are the square roots of the eigenvalues of $\rho\cdot\tilde{\rho}$
in decreasing order, and $\tilde{\rho}$ is defined as $\tilde{\rho}=\left(\sigma_{y}\otimes\sigma_{y}\right)\rho^{\star}\left(\sigma_{y}\otimes\sigma_{y}\right)$. 

For most quantum states no exact expression for any measure of entanglement
is known, and thus numerical algorithms must be used. One of the first
algorithms computing entanglement has been presented in \cite{.Zyczkowski1999}.
There the entanglement of formation was approximated using a random
walk algorithm on the space of the decompositions of the given mixed
state. A much faster algorithm for the entanglement of formation was
presented in \cite{Audenaert2001}. This algorithm made use of the
conjugate gradient method. In \cite{Rothlisberger2009} the authors
extended and improved the algorithm. The authors also applied the
algorithm to the convex roof extension of the multipartite Meyer-Wallach
measure \cite{Meyer2002}. 

In this paper we present an algorithm for the geometric measure of
entanglement. The algorithm is easy to implement, since every step
is either the solution of an eigenproblem or finding a singular value
decomposition of a matrix, and no further numerical techniques are
needed.

This paper is organized as follows. In Section \ref{sec:Algorithm}
we present the algorithm for pure and mixed states. We also discuss
its properties and convergence. In Section \ref{sec:Applications}
we test our algorithm on bipartite and multipartite mixed states with
known value of the geometric measure of entanglement. Further we compute
an approximation of the geometric measure of entanglement for the
isotropic states of three qubits, and the three-qubit XX model with
a constant magnetic field. We conclude in Section \ref{sec:Concluding-remarks}.

\section{\label{sec:Algorithm}Algorithm}

Before we present our algorithm for general multipartite states, we
begin with bipartite and multipartite pure states.

\subsection{Pure states}

\subsubsection{\label{sub:Bipartite-states}Bipartite states}

For bipartite pure states $\ket{\psi}\in\mathcal{H}_{1}\otimes\mathcal{H}_{2}$
the geometric measure of entanglement is given by \cite{Shimony1995}\begin{equation}
E_{G}\left(\ket{\psi}\right)=1-\lambda_{max}^{2},\end{equation}
where $\lambda_{max}$ is the largest Schmidt coefficient of $\ket{\psi}$.
Note that $\lambda_{max}^{2}$ is also the maximal eigenvalue of $\mathrm{Tr}_{1}\left[\ket{\psi}\bra{\psi}\right]$
and $\mathrm{Tr}_{2}\left[\ket{\psi}\bra{\psi}\right]$. Further let
$\ket{\phi_{1}}\in\mathcal{H}_{1}$ and $\ket{\phi_{2}}\in\mathcal{H}_{2}$
be the eigenstate corresponding to the maximal eigenvalue of $\mathrm{Tr}_{2}\left[\ket{\psi}\bra{\psi}\right]$,
and $\mathrm{Tr}_{1}\left[\ket{\psi}\bra{\psi}\right]$ respectively.
Then the state $\ket{\phi}=\ket{\phi_{1}}\otimes\ket{\phi_{2}}$ is
a closest separable state to $\ket{\psi}$.

\subsubsection{\label{sub:Multipartite-states}Multipartite states}

If we consider pure states $\ket{\psi}$ on an $n$-partite Hilbert
space $\mathcal{H}\in\otimes_{i=1}^{n}\mathcal{H}_{i}$ with $n>2$,
the geometric measure of entanglement is only known for a few special
cases \cite{Wei2003,Tamaryan2008}. In \cite{Shimoni2005,Most2010}
the authors presented an algorithm for an approximation of $E_{G}$
for pure states. For simplicity we will discuss the algorithm from
\cite{Shimoni2005,Most2010} for a pure state of three qubits, a generalization
to arbitrary systems is done in the end of this section.

Let $\ket{\psi}$ be the given state of three qubits. The algorithm
starts with a random product state $\ket{\phi_{0}}=\ket{0_{0}^{\left(1\right)}}\ket{0_{0}^{\left(2\right)}}\ket{0_{0}^{\left(3\right)}}$
of three qubits, where the lower index will be used for counting the
steps of the algorithm and the upper index denotes the {}``number''
of the qubit. Now we consider $\ket{\tilde{\psi}}=\left(\bra{0_{0}^{\left(2\right)}}\bra{0_{0}^{\left(3\right)}}\right)\ket{\psi}$,
which is a pure unnormalized state on the space of the first qubit.
If we want to maximize the overlap $\left|\braket{\phi_{0}|\psi}\right|$
for fixed states $\ket{0_{0}^{\left(2\right)}}$ and $\ket{0_{0}^{\left(3\right)}}$,
we have to replace $\ket{0_{0}^{\left(1\right)}}$ by the state $\ket{0_{1}^{\left(1\right)}}=\frac{1}{\sqrt{\braket{\tilde{\psi}|\tilde{\psi}}}}\ket{\tilde{\psi}}$.
The procedure is repeated for the second qubit, starting in the product
state $\ket{0_{1}^{\left(1\right)}}\ket{0_{0}^{\left(2\right)}}\ket{0_{0}^{\left(3\right)}}$
and resulting in the state $\ket{0_{1}^{\left(1\right)}}\ket{0_{1}^{\left(2\right)}}\ket{0_{0}^{\left(3\right)}}$.
Finally, the same maximization is done for the third qubit with the
final state $\ket{\phi_{1}}=\ket{0_{1}^{\left(1\right)}}\ket{0_{1}^{\left(2\right)}}\ket{0_{1}^{\left(3\right)}}$.
In the same way we define the product state $\ket{\phi_{n}}=\ket{0_{n}^{\left(1\right)}}\ket{0_{n}^{\left(2\right)}}\ket{0_{n}^{\left(3\right)}}$
to be the result of $n$ iterations of the algorithm. In the following
we will prove some properties of the algorithm.
\begin{prop}
\label{pro:limit}Let $\ket{000}=\lim_{n\rightarrow\infty}\ket{\phi_{n}}$
be the product state after an infinite number of steps of the algorithm,
then holds: \begin{equation}
\braket{100|\psi}=\braket{010|\psi}=\braket{001|\psi}=0.\label{eq:limit}\end{equation}
\end{prop}
\begin{proof}
If $\braket{100|\psi}\neq0$, then there exists a product state of
the form $ $$\ket{\phi}=\ket{\phi^{\left(1\right)}}\ket{00}$ such
that $\left|\braket{\phi|\psi}\right|>\left|\braket{000|\psi}\right|$.
This means that $\ket{000}\neq\lim_{n\rightarrow\infty}\ket{\phi_{n}}$,
which is a contradiction to the definition of $\ket{000}$. Using
the same argument it can be seen that $\braket{010|\psi}=\braket{001|\psi}=0$
also holds.
\end{proof}
From Proposition \ref{pro:limit} we see that the state $\ket{\psi}$
can be written as follows: \begin{equation}
\ket{\psi}=\lambda_{1}\ket{000}+\lambda_{2}\ket{110}+\lambda_{3}\ket{101}+\lambda_{4}\ket{011}+\lambda_{5}\ket{111},\label{eq:generalized schmidt}\end{equation}
where four of the coefficients $\lambda_{i}$ can be chosen real and
nonnegative, and $\sum_{i}\left|\lambda_{i}\right|^{2}=1$. The form
(\ref{eq:generalized schmidt}) is also known as generalized Schmidt
decomposition \cite{Ac'in2000,Carteret2000}.
\begin{prop}
The algorithm computes a generalized Schmidt decomposition of a pure
state with an arbitrary given precision.\end{prop}
\begin{proof}
In order to find a generalized Schmidt decomposition with a given
precision $\varepsilon$ we need to find five parameters $\mu_{i}$
with $\sum_{i=1}^{5}\left|\mu_{i}\right|^{2}=1$ and a product basis
$\left\{ \ket{ijk}\right\} $ such that the state \begin{equation}
\ket{\psi_{\mathrm{approx}}}=\mu_{1}\ket{000}+\mu_{2}\ket{110}+\mu_{3}\ket{101}+\mu_{4}\ket{011}+\mu_{5}\ket{111}\end{equation}
is closer to $\ket{\psi}$ than $\varepsilon$, i.e. $D\left(\ket{\psi},\ket{\psi_{\mathrm{approx}}}\right)\leq\varepsilon$
with the trace distance $D\left(\ket{\psi},\ket{\phi}\right)=\sqrt{1-\left|\braket{\psi|\phi}\right|^{2}}$.
This is accomplished by the state \begin{equation}
\ket{\psi_{n}}=\frac{1}{N}\sum_{i,j,k}b_{ijk}\ket{ijk}_{n},\label{eq:psi_n}\end{equation}
where $\ket{ijk}_{n}=\ket{i_{n}^{\left(1\right)}}\ket{j_{n}^{\left(2\right)}}\ket{k_{n}^{\left(3\right)}}$
are the basis states after $n$ iterations of the algorithm. The coefficients
$b_{ijk}$ are defined as follows: $b_{100}=b_{010}=b_{001}=0$, and
$b_{ijk}=\left(\braket{\psi|ijk}_{n}\right)^{\star}$ otherwise. $N$
assures normalization of $\ket{\psi_{n}}$. The trace distance between
$\ket{\psi}$ and $\ket{\psi_{n}}$ becomes $D\left(\ket{\psi},\ket{\psi_{n}}\right)=\sqrt{\left|\braket{\psi|100}_{n}\right|^{2}+\left|\braket{\psi|010}_{n}\right|^{2}+\left|\braket{\psi|001}_{n}\right|^{2}}$.
Using Proposition \ref{pro:limit} we see that $\lim_{n\rightarrow\infty}D\left(\ket{\psi},\ket{\psi_{n}}\right)=0$.
The wanted approximation $\ket{\psi_{\mathrm{approx}}}$ is obtained
by a state $\ket{\psi_{n}}$ such that $D\left(\ket{\psi},\ket{\psi_{n}}\right)\leq\varepsilon$.
\end{proof}
Thus we showed, that the algorithm presented in the beginning of this
section computes a generalized Schmidt decomposition of the given
pure state. As the generalized Schmidt decomposition is in general
not unique \cite{Ac'in2000,Carteret2000}, the result of the computation
may depend on the choice of the initial product state $\ket{\phi_{0}}$.
In particular, the final overlap $1-\left|\braket{000|\psi}\right|^{2}$
does not have to be the geometric measure of entanglement, even for
an infinite number of iterations.

Finally we note that all results presented in this section can be
extended to an arbitrary number of qubits. Then the equations have
to be changed accordingly. For four qubits, eq. (\ref{eq:limit})
becomes $\braket{1000|\psi}=\braket{0100|\psi}=\braket{0010|\psi}=\braket{0001|\psi}=0$.
Moreover the results even hold if the subsystems are not qubits, but
have arbitrary dimensions. For simplicity we consider a pure state
of three qutrits in the following. Again $\ket{000}=\lim_{n\rightarrow\infty}\ket{\phi_{n}}$
denotes the product state which is achieved after infinite number
of iterations. Using the same arguments as in the proof of Proposition
\ref{pro:limit} we see: \begin{eqnarray}
\braket{100|\psi} & = & \braket{010|\psi}=\braket{001|\psi}=0,\\
\braket{200|\psi} & = & \braket{020|\psi}=\braket{002|\psi}=0,\end{eqnarray}
where $\ket{1}$ and $\ket{2}$ are arbitrary states orthogonal to
$\ket{0}$ on the corresponding subspace. In order to find a generalized
Schmidt decomposition we also have to find specific states $\ket{1}$
and $\ket{2}$ for each subspace. Let $\ket{\psi}=\sum_{i=0}^{2}\sum_{j=0}^{2}\sum_{j=0}^{2}a_{ijk}\ket{ijk}$
be the expansion of the state in a product basis containing $\ket{000}$.
Then consider the unnormalized state $\ket{\tilde{\psi}}=\sum_{i=1}^{2}\sum_{j=1}^{2}\sum_{j=1}^{2}a_{ijk}\ket{ijk}$.
Since in the present stage of the algorithm we only have the knowledge
about the state $\ket{000}=\ket{0^{\left(1\right)}}\ket{0^{\left(2\right)}}\ket{0^{\left(3\right)}}$,
the state $\ket{\tilde{\psi}}$ can be computed as follows. Starting
from the state $\ket{\psi}$ we compute the unnormalized state $\ket{\alpha}=\ket{\psi}-\ket{000}\braket{000|\psi}$.
In the second step we compute $\ket{\beta}=\ket{\alpha}-\sum_{i<j}\ket{0^{\left(i\right)}0^{\left(j\right)}}\braket{0^{\left(i\right)}0^{\left(j\right)}|\alpha}$.
In the final step we get $\ket{\tilde{\psi}}=\ket{\beta}-\sum_{i}\ket{0^{\left(i\right)}}\braket{0^{\left(i\right)}|\beta}$.
The state $\ket{\tilde{\psi}}$ is an unnormalized pure state of three
qubits, and according to Proposition \ref{pro:limit} applying the
algorithm to it will give us the desired product basis $\left\{ \ket{ijk}\right\} $
with the property $\braket{211|\psi}=\braket{121|\psi}=\braket{112|\psi}=0$.
The expansion of the state $\ket{\psi}$ in the final product basis
$\left\{ \ket{ijk}\right\} $ is a generalized Schmidt decomposition
of $\ket{\psi}$ \cite{Carteret2000}. Let $\left\{ \ket{ijk}_{n}\right\} $
be the computed product basis after $n$ iterations of the algorithm.
The approximated generalized Schmidt decomposition of $\ket{\psi}$
becomes \begin{equation}
\ket{\psi_{n}}=\frac{1}{N}\sum_{i,j,k}b_{ijk}\ket{ijk}_{n}\end{equation}
with $b_{iij}=b_{iji}=b_{jii}=0$ for $i<j$ and $b_{ijk}=a_{ijk}$
otherwise. $N$ assures normalization of $\ket{\psi_{n}}$. The precision
of the approximation is then given by $D\left(\ket{\psi},\ket{\psi_{n}}\right)=\sqrt{\sum_{i<j}\left(\left|\braket{iij|\psi}\right|^{2}+\left|\braket{iji|\psi}\right|^{2}+\left|\braket{jii|\psi}\right|^{2}\right)}$.
In the same way we can find a generalized Schmidt decomposition for
any multipartite pure state with an arbitrary precision.

\subsection{\label{sub:Mixed-states}Mixed states}

The main idea of the algorithm for mixed states is a consequence of
the fact, that the geometric measure of entanglement may also be written
as \cite{Streltsov2010}\begin{equation}
E_{G}\left(\rho\right)=1-\max_{\sigma\in S}F\left(\rho,\sigma\right),\label{eq:Eg-1}\end{equation}
where $S$ denotes the set of separable states and $F\left(\rho,\sigma\right)=\left(\mathrm{Tr}\left[\sqrt{\sqrt{\rho}\sigma\sqrt{\rho}}\right]\right)^{2}$
is the fidelity. Let $\ket{\psi}\in\mathcal{H}\otimes\mathcal{H}_{a}$
be a purification of $\rho$. It can be written as \begin{equation}
\ket{\psi}=\sum_{i}\sqrt{p_{i}}\ket{\psi_{i}}\otimes\ket{i},\label{eq:psi}\end{equation}
with probabilities $p_{i}$ and $\rho=\sum_{i}p_{i}\ket{\psi_{i}}\bra{\psi_{i}}$.
According to Uhlmann's theorem \cite[page 410]{Nielsen2000} and using
(\ref{eq:Eg-1}) we can also write \begin{equation}
E_{G}\left(\rho\right)=1-\max_{\mathrm{Tr}_{a}\left[\ket{\phi}\bra{\phi}\right]\in S}\left|\braket{\psi|\phi}\right|^{2},\label{eq:Eg-2}\end{equation}
where the maximization is done over all states $\ket{\phi}\in\mathcal{H}\otimes\mathcal{H}_{a}$
which are purifications of a separable state. Note that any $\ket{\phi}$
can be written in the form \begin{equation}
\ket{\phi}=\sum_{j}\sqrt{q_{j}}\ket{\phi_{j}}\otimes U^{\dagger}\ket{j},\label{eq:phi}\end{equation}
with pure separable states $\ket{\phi_{j}}\in S$, probabilities $q_{j}$,
a unitary $U$ acting on the Hilbert space $\mathcal{H}_{a}$, and
$\braket{i|j}=\delta_{ij}$. 

From (\ref{eq:Eg-2}) we see, that we can get an approximation of
$E_{G}$ by maximizing the overlap $\left|\braket{\psi|\phi}\right|$
over all states $\ket{\phi}$ of the form (\ref{eq:phi}). Our approach
for this maximization is the following:
\begin{enumerate}
\item For fixed $q_{i}$ and $\ket{\phi_{i}}$ we find a unitary $U$ in
(\ref{eq:phi}) such that the overlap $\left|\braket{\psi|\phi}\right|$
is maximal.
\item For fixed $U$ and $q_{i}$ we find states $\ket{\phi_{i}}$ in (\ref{eq:phi})
such that the overlap $\left|\braket{\psi|\phi}\right|$ is maximal.
Note that this is in general only possible for bipartite states. For
multipartite states we compute $\ket{\phi_{i}}$ such that the overlap
$\left|\braket{\psi|\phi}\right|$ does not decrease. 
\item For fixed $U$ and $\ket{\phi_{i}}$ we find probabilities $q_{i}$
in (\ref{eq:phi}) such that the overlap $\left|\braket{\psi|\phi}\right|$
is maximal.
\end{enumerate}
Steps 1-3 are iterated until the increase of the overlap $\left|\braket{\psi|\phi}\right|$
is smaller than a small parameter $\varepsilon>0$. When the algorithm
stops, the approximation of the geometric measure of entanglement
is given by $\tilde{E}_{G}\left(\rho\right)=1-\left|\braket{\psi|\tilde{\phi}}\right|^{2}$,
where $\ket{\tilde{\phi}}$ is the final state of the form (\ref{eq:phi}).

In the following section we will discuss the properties of the algorithm.
Note that the order of the steps presented above can also be changed
without changing these properties.

\subsection{Properties}

In the following we will discuss some properties of the algorithm
presented above. In the first step the probabilities $q_{i}$ and
the separable pure states $\ket{\phi_{i}}$ are fixed. The product
$\left|\braket{\psi|\phi}\right|$ can be maximized using Uhlmann's
theorem \cite[page 410]{Nielsen2000}, it is maximal if $U$ is chosen
such that holds: \begin{equation}
A=\sqrt{AA^{\dagger}}U^{\dagger},\label{eq:A}\end{equation}
where $A$ is a matrix defined as $A=\sum_{i,j}\sqrt{p_{i}q_{j}}\braket{\phi_{j}|\psi_{i}}\ket{i}\bra{j}$.
Note that equation (\ref{eq:A}) is the polar decomposition of $A$,
which can be computed efficiently for any matrix $A$ \cite{Press2007}. 

In the second step of the algorithm we fix $U$, which was found in
the step before. The probabilities $q_{i}$ are also unchanged. In
order to maximize the overlap $\left|\braket{\psi|\phi}\right|$ the
separable states $\ket{\phi_{i}}$ have to be changed to the states
$\ket{\phi_{i}'}$ for which holds: \begin{equation}
\braket{\psi_{i}'|\phi_{i}'}=\sqrt{F_{s}\left(\ket{\psi_{i}'}\right)},\label{eq:psi'}\end{equation}
with the states $\ket{\psi_{i}'}=\frac{1}{\sqrt{p_{i}'}}\sum_{j}u_{ij}\sqrt{p_{j}}\ket{\psi_{j}}$,
where $u_{ij}=\braket{i|U|j}$ are elements of $U$ in the computational
basis, and $p_{i}'>0$ is chosen such that $\ket{\psi_{i}'}$ is normalized.
For bipartite states $\ket{\psi_{i}'}$ this step is evaluated according
to the discussion in Section \ref{sub:Bipartite-states}. If $\ket{\psi'_{i}}$
is multipartite, the closest separable state $\ket{\phi'_{i}}$ cannot
be found in general. However, there is a way to circumvent this problem
as follows. We apply the algorithm described in Section \ref{sub:Multipartite-states}
to the state $\ket{\psi_{i}'}$ with the initial product state $\ket{\phi_{i}}$,
thus getting a final product state $\ket{\phi_{i}'}$. The state $\ket{\phi_{i}'}$
is not necessarily the closest separable state to $\ket{\psi_{i}'}$,
however it will be closer to $\ket{\psi_{i}'}$ than the initial product
state $\ket{\phi_{i}}$. But then, if we replace $\ket{\phi_{i}}$
by $\ket{\phi_{i}'}$, we get a better approximation of the geometric
measure of entanglement. This can be seen by noting that for the overlaps
of the purifications holds: $\left|\braket{\psi|\phi'}\right|\geq\left|\braket{\psi|\phi}\right|$,
where in $\ket{\phi'}$ all product states $\ket{\phi_{i}}$ were
replaced by $\ket{\phi_{i}'}$.

In the last step of the iteration we fix $U$ which was found in the
first step, and the separable states $\ket{\phi_{i}'}$ which were
found in the second step. Using the method of Lagrange multipliers
we find the optimal probabilities: \begin{equation}
q_{i}'=\frac{p_{i}'\left|\braket{\psi_{i}'|\phi_{i}'}\right|^{2}}{\sum_{k}p_{k}'\left|\braket{\psi_{k}'|\phi_{k}'}\right|^{2}}.\end{equation}

Let $\tilde{E}_{n}\left(\rho\right)$ be the approximation of the
geometric measure of entanglement after $n$ iterations of the algorithm.
We will now prove the main property of the algorithm.
\begin{prop}
The approximated value of the geometric measure of entanglement never
increases in a step of the iteration: \begin{equation}
\tilde{E}_{n+1}\left(\rho\right)\leq\tilde{E}_{n}\left(\rho\right).\end{equation}
\end{prop}
\begin{proof}
It is sufficient to show that the overlap of the purifications $\left|\braket{\psi|\phi}\right|$
does not decrease in any step of the algorithm. This is seen directly
from the definition of the algorithm in Section \ref{sub:Mixed-states}.
\end{proof}

\subsection{\label{sub:Implementation}Implementation}

First we set a small parameter $\varepsilon>0$. The algorithm starts
with a random decomposition $\left\{ p_{i},\ket{\psi_{i}}\right\} _{i=1}^{d^{2}}$
into $d^{2}$ elements of the state $ $$\rho=\sum_{i=1}^{d^{2}}p_{i}\ket{\psi_{i}}\bra{\psi_{i}}$
and a separable decomposition $\left\{ q_{i},\ket{\phi_{i}}\right\} _{i=1}^{d^{2}}$
of a random separable state $\sigma=\sum_{i=1}^{d^{2}}q_{i}\ket{\phi_{i}}\bra{\phi_{i}}$,
where we demand that $p_{i}>0$ and $q_{i}>0$ for all $1\leq i\leq d^{2}$.
The steps 1-3 from the Section \ref{sub:Mixed-states} can be implemented
as follows:
\begin{enumerate}
\item Find the singular value decomposition of the matrix $A=\sum_{i,j}\sqrt{p_{i}q_{j}}\braket{\phi_{j}|\psi_{i}}\ket{i}\bra{j}$,
i.e. $A=VDW$ with unitary matrices $V$, $W$ and diagonal nonnegative
matrix $D$. Define $U=W^{\dagger}V^{\dagger}$, noting that (\ref{eq:A})
is fulfilled.
\item Define unnormalized states \begin{equation}
\ket{\alpha_{i}}=\sum_{j=1}^{d^{2}}u_{ij}\sqrt{p_{j}}\ket{\psi_{j}},\end{equation}
with $u_{ij}=\braket{i|U|j}$. Compute $p_{i}'=\braket{\alpha_{i}|\alpha_{i}}$
and $\ket{\psi_{i}'}=\frac{1}{\sqrt{p_{i}'}}\ket{\alpha_{i}}$ for
all $i$. For bipartite states compute separable pure states $\ket{\phi_{i}'}\in S$
such that $\braket{\psi_{i}'|\phi_{i}'}=\sqrt{F_{s}\left(\ket{\psi_{i}'}\right)}$
. For multipartite states find product states $\ket{\phi_{i}'}$ which
are closer to $\ket{\psi_{i}'}$ than the states $\ket{\phi_{i}}$
computed in the step before. This can be done applying the algorithm
presented in Section \ref{sub:Multipartite-states} to the state $\ket{\psi_{i}'}$
with the initial product state $\ket{\phi_{i}}$.
\item Compute $q_{i}'=\frac{p_{i}'\left|\braket{\psi_{i}'|\phi_{i}'}\right|^{2}}{\sum_{k}p_{k}'\left|\braket{\psi_{k}'|\phi_{k}'}\right|^{2}}$. 
\end{enumerate}
After performing steps 1-3 define a new separable state $\sigma'=\sum_{i}q_{i}'\ket{\phi_{i}'}\bra{\phi_{i}'}$,
which is an approximation of the closest separable state to $\rho$.
If $F\left(\rho,\sigma'\right)-F\left(\rho,\sigma\right)>\varepsilon$,
set $\ket{\psi_{i}}=\ket{\psi_{i}'}$, $\ket{\phi_{i}}=\ket{\phi_{i}'}$,
$p_{i}=p_{i}'$ and $q_{i}=q_{i}'$ for all $i$ and go back to step
1, otherwise stop. The computed approximation is $\tilde{E}_{G}\left(\rho\right)=1-F\left(\rho,\sigma'\right)$.

\subsection{\label{sub:Convergence}Convergence}

One of the most important questions regarding algorithms computing
entanglement is whether or not the algorithm converges to the exact
value of the entanglement measure, at least for infinite number of
steps. For a general multipartite state with more than two parties
the algorithm will converge to the wrong value with some nonzero probability,
depending on the initial separable state. This is due to the fact,
that the algorithm for pure multipartite states presented in Section
\ref{sub:Multipartite-states} does not necessarily compute the correct
value \cite{Most2007,Most2010}. 

For bipartite mixed states there is no full answer to this question,
and testing the algorithm on bipartite states with known geometric
measure of entanglement we did not observe convergence to a wrong
value. However it can be shown that for some states and some special
choice of the purifications $\ket{\psi}$ and $\ket{\phi}$ the algorithm
does not compute the correct value even after an infinite number of
iterations. To see this we consider a separable state $\rho\in S$
with rank $r$ such that any separable decomposition of $\rho$ has
more elements than $r$. The existence of such states is assured \cite{Horodecki2009}.
Let now $\left\{ p_{i},\ket{\psi_{i}}\right\} _{i=1}^{r}$ be a decomposition
of $\rho$ which is optimal among all decompositions with $r$ elements,
i.e. the average entanglement $\sum_{i=1}^{r}p_{i}E_{G}\left(\ket{\psi_{i}}\right)$
is minimal among all decompositions into $r$ elements. Further let
$\ket{\phi_{i}}$ be the closest separable state to $\ket{\psi_{i}}$
and we also choose $q_{i}=\frac{p_{i}\left|\braket{\psi_{i}|\phi_{i}}\right|^{2}}{\sum_{k}p_{k}\left|\braket{\psi_{k}|\phi_{k}}\right|^{2}}$.
Now we start the algorithm with the decompositions $\left\{ p_{i},\ket{\psi_{i}}\right\} _{i=1}^{r}$
and $\left\{ q_{i},\ket{\phi_{i}}\right\} _{i=1}^{r}$, as described
in the previous section. Then the unitary $U$ which maximizes the
overlap of the purifications $\ket{\psi}=\sum_{i}\sqrt{p_{i}}\ket{\psi_{i}}\otimes\ket{i}$
and $\ket{\phi}=\sum_{j}\sqrt{q_{j}}\ket{\phi_{j}}\otimes U^{\dagger}\ket{j}$
is given by $U=\openone$. In the second step the algorithm will maximize
the overlaps $\braket{\phi_{i}|\psi_{i}}$, which are already optimal.
The same is true for the last step of the algorithm, where the probabilities
$q_{j}$ are optimized. Thus the algorithm preserves the initial separable
state, and does not compute the correct value even for infinite number
of steps.

To avoid the problem mentioned above the algorithm should always start
with a separable state chosen at random, i.e. with random initial
probabilities $q_{i}$ and random separable pure states $\ket{\phi_{i}}$.
Moreover, the number of initial nonzero probabilities $q_{i}$ should
be at least $\left(\dim\mathcal{H}\right)^{2}$.

In the following section we will test the algorithm and present some
applications for states with unknown geometric measure of entanglement.

\section{\label{sec:Applications}Applications}

\subsection{\label{sub:Testing-the-algorithm}Testing the algorithm}

\subsubsection{\label{sub:Two-qubits}Two qubits}

If $\rho$ is a two-qubit state, the geometric measure of entanglement
is given by (\ref{eq:twoqubits}). We applied our algorithm with $\varepsilon=10^{-15}$
to $10^{3}$ random states of two qubits and tested the computed value
$\tilde{E}_{G}$ against the exact value given in (\ref{eq:twoqubits}).
The maximal deviation $\tilde{E}_{G}-E_{G}$ from the exact value
was $6\cdot10^{-11}$. The average number of steps made by the algorithm
was 291.

\subsubsection{Isotropic states}

We also tested our algorithm on the isotropic states in dimension
$d\times d$, these are states of the form \begin{equation}
\rho=p\ket{\Phi^{+}}\bra{\Phi^{+}}+\frac{1-p}{d^{2}}\openone,\label{eq:iso-1}\end{equation}
with the maximally entangled state $\ket{\Phi^{+}}=\frac{1}{\sqrt{d}}\sum_{i=0}^{d}\ket{ii}$.
For these states an exact expression for the geometric measure of
entanglement was given in \cite{Wei2003}, the states are entangled
if and only if $p>\frac{1}{1+d}$. We applied our algorithm to the
state (\ref{eq:iso-1}) for $2\le d\leq3$ with the parameter $\varepsilon=10^{-15}$
for $p=0.01n$ and $0\leq n\leq99$. The difference between the approximated
value $\tilde{E}_{G}$ and the exact value $E_{G}$ was always less
than $10^{-10}$. 

In order to do the test for $d=4$ within a reasonable time some modifications
had to be applied. First, we minimized only over decompositions into
$d^{2}=16$ instead of $d^{4}=256$ pure states. Further, for $d=4$
the test was done on entangled states only, i.e. for $p=0.01n$ with
$20<n\leq99$. The difference between the approximation $\tilde{E}_{G}$
and the exact value $E_{G}$ never exceeded $10^{-13}$. The results
are summarized in Table \ref{tab:isotropic}. There $\bar{N}$ denotes
the average number of steps made by the algorithm. %
\begin{table}[H]
\begin{centering}
\begin{tabular}{c|c|c|c}
$d$ & $2$ & $3$ & $4$\tabularnewline
\hline 
$\tilde{E}_{G}-E_{G}$ & $<10^{-13}$ & $<10^{-10}$ & $<10^{-13}$\tabularnewline
\hline 
$\bar{N}$ & $80$ & $516$ & $2259$\tabularnewline
\end{tabular}
\par\end{centering}

\caption{\label{tab:isotropic} Precision of the approximation $\tilde{E}_{G}-E_{G}$
and the average number of steps $\bar{N}$ for the isotropic states
(\ref{eq:iso-1}) with parameter $\varepsilon=10^{-15}$.}

\end{table}

For the cases tested above the algorithm always converged into the
correct value of $E_{G}$ within the precision given in Table \ref{tab:isotropic}
with a single run of the algorithm. Note that in general more than
one run with different initial parameters should be done to avoid
convergence into a wrong value. Further we see from Table \ref{tab:isotropic}
that the parameter $\varepsilon$ should not be used directly to quantify
the precision of the approximation, although the deviation from the
exact value is very small.

\subsubsection{Four qubits}

In \cite{Guhne2008a} the authors computed the geometric measure of
entanglement for a class of mixed states of four qubits. We tested
our algorithm on the state $\rho\left(t\right)$, which for $t=0$
is defined as the four-qubit cluster state \begin{equation}
\ket{\mathrm{CL}_{4}}=\frac{1}{2}\left(\ket{0000}+\ket{0011}+\ket{1100}-\ket{1111}\right).\end{equation}
For $t>0$ the diagonal terms of $\rho$ are left invariant, and the
off-diagonal components decay exponentially with $t$, i.e. \begin{equation}
\rho_{kl}\left(t\right)=\begin{cases}
\rho_{kl}\left(0\right) & \textrm{for }k=l,\\
e^{-t}\rho_{kl}\left(0\right) & \textrm{for }k\neq l.\end{cases}\label{eq:guhne}\end{equation}
We applied our algorithm with parameter $\varepsilon=10^{-15}$ on
the states $\rho\left(t\right)$ with $t=0.01n$ for all $1\leq n\leq100$.
The discrepancy between the approximated value and the exact value
given in \cite{Guhne2008a} was always smaller than $10^{-14}$. 

The same test was done for the state $\tilde{\rho}\left(t\right)$,
which for $t=0$ is defined as the four-qubit $W$ state \[
\ket{W_{4}}=\frac{1}{2}\left(\ket{0001}+\ket{0010}+\ket{0100}+\ket{1000}\right),\]
and for $t>0$ the off-diagonal components decay exponentially as
given in (\ref{eq:guhne}). There the discrepancy between the approximation
and the exact value was always smaller than $10^{-11}$.

Finally we tested our algorithm on the four-qubit state $\bar{\rho}\left(t\right)$,
which for $t=0$ is defined as the symmetrized Dicke state \begin{eqnarray}
\ket{D_{4}} & = & \frac{1}{\sqrt{6}}\left(\ket{0011}+\ket{0101}+\ket{1001}+\ket{1100}+\ket{0110}\right.\nonumber \\
 &  & +\left.\ket{1010}\right).\end{eqnarray}
Again for $t>0$ the off-diagonal components decay as in (\ref{eq:guhne}).
The test was done with $t=0.01n$ for all $1\leq n\leq100$, the difference
$\tilde{E}_{G}-E_{G}$ was always smaller than $10^{-12}$. The results
are summarized in Table \ref{tab:four qubits}. There $\bar{N}$ denotes
the average number of iterations made by the algorithm. %
\begin{table}[H]
\begin{centering}
\begin{tabular}{c|c|c|c}
$\rho\left(0\right)$ & $\ket{\mathrm{CL}_{4}}$ & $\ket{W_{4}}$ & $\ket{D_{4}}$\tabularnewline
\hline 
$\tilde{E}_{G}-E_{G}$ & $<10^{-14}$ & $<10^{-11}$ & $<10^{-12}$\tabularnewline
\hline 
$\bar{N}$ & $12$ & $173$ & $126$\tabularnewline
\end{tabular}
\par\end{centering}

\caption{\label{tab:four qubits} Precision of the approximation $\tilde{E}_{G}-E_{G}$
and the average number of steps $\bar{N}$ for the four-qubit states
presented in the text with parameter $\varepsilon=10^{-15}$.}

\end{table}

Note that the optimizations above were done over pure state decompositions
into $2^{4}$ elements instead of $2^{8}$. This reduction was needed
in order to do the computation within a reasonable time. Moreover
we note that for very small parameter $t=0.01$ we sometimes observed
convergence into a wrong value. This is due to the fact that for small
$t$ the state $\rho\left(t\right)$ is almost pure. As was mentioned
in Section \ref{sub:Convergence} the algorithm can converge to wrong
values for pure multipartite states. In these cases the algorithm
was started again with random initial parameters. To get an impression
we mention that for the last example $\bar{\rho}\left(0.01\right)$
the algorithm sometimes converged to $\tilde{E}_{G}-E_{G}\approx8\cdot10^{-4}$.

\subsubsection{Comparison with other algorithms}

A significant difference between our algorithm and the algorithms
presented in \cite{Audenaert2001,Rothlisberger2009} is the fact,
that our algorithm implies only the solution of the eigenproblem and
finding a singular value decomposition. For both problems efficient
numerical algorithms exist \cite{Press2007}, implying that each step
of our algorithm can be done efficiently. The algorithms based on
conjugate gradients usually imply a line search \cite{Audenaert2001}.
It is not known to us whether a line search can in general be done
efficiently for the problem considered here.

As noted in Section \ref{sub:Two-qubits}, the average number of iterations
made by our algorithm for random two-qubit states with parameter $\varepsilon=10^{-15}$
was 291. This is comparable to the performance of the conjugate gradient
algorithm, for comparison see Figure 1 in \cite{Rothlisberger2009}.

\subsection{On additivity of entanglement}

A measure of entanglement $E$ is called additive, if for any two
states $\rho^{AB}$ and $\sigma^{AB}$ holds \cite{Plenio2007}:\begin{equation}
E\left(\rho^{AB}\otimes\sigma^{AB}\right)=E\left(\rho^{AB}\right)+E\left(\sigma^{AB}\right),\end{equation}
where the entanglement between the parties $A$ and $B$ is considered.

For pure states $ $$\ket{\psi^{AB}}$ and $\ket{\phi^{AB}}$ we see
that \begin{equation}
F_{s}\left(\ket{\psi^{AB}}\otimes\ket{\phi^{AB}}\right)=F_{s}\left(\ket{\psi^{AB}}\right)F_{s}\left(\ket{\phi^{AB}}\right),\label{eq:additivity}\end{equation}
with $F_{s}\left(\rho\right)=\underset{\sigma\in S}{\max}F\left(\rho,\sigma\right)$
and the fidelity $F\left(\rho,\sigma\right)=\left(\mathrm{Tr}\left[\sqrt{\sqrt{\rho}\sigma\sqrt{\rho}}\right]\right)^{2}$.
From (\ref{eq:additivity}) we see that the geometric measure of entanglement
is not additive. Note, that for the entanglement of formation nonadditivity
has also been proved \cite{Hastings2009}.

We consider the \emph{logarithmic entanglement} \begin{equation}
E_{log}\left(\rho\right)=-\log_{2}F_{s}\left(\rho\right),\end{equation}
which is additive for pure bipartite states, as is seen from (\ref{eq:additivity}).
In general holds: $F_{s}\left(\rho^{AB}\otimes\sigma^{AB}\right)\geq F_{s}\left(\rho^{AB}\right)F_{s}\left(\sigma^{AB}\right)$,
and thus the logarithmic entanglement is subadditive: \begin{equation}
E_{log}\left(\rho^{AB}\otimes\sigma^{AB}\right)\leq E_{log}\left(\rho^{AB}\right)+E_{log}\left(\sigma^{AB}\right).\label{eq:additivity-1}\end{equation}
We use our algorithm to test the inequality (\ref{eq:additivity-1}).
Note that for two-qubit states $\rho$ we get $F_{s}\left(\rho\right)=\frac{1}{2}\left(1+\sqrt{1-C\left(\rho\right)^{2}}\right)$.
We take $\rho^{AB}$ and $\sigma^{AB}$ to be random states of two
qubits, and apply the algorithm to $\rho^{AB}\otimes\sigma^{AB}$
with parameter $\varepsilon=10^{-7}$. This procedure is repeated
100 times, each time the computed approximation $\tilde{F}_{s}\left(\rho^{AB}\otimes\sigma^{AB}\right)$
was slightly below $F_{s}\left(\rho^{AB}\right)F_{s}\left(\sigma^{AB}\right)$,
which means that we could not disprove additivity of logarithmic entanglement
in this way. The difference $F_{s}\left(\rho^{AB}\right)F_{s}\left(\sigma^{AB}\right)-\tilde{F}_{s}\left(\rho^{AB}\otimes\sigma^{AB}\right)$
was always smaller than $10^{-5}$. \[
\]

\subsection{Applications to 3 qubits}

\begin{figure}
\noindent \begin{centering}
\includegraphics[width=1\columnwidth]{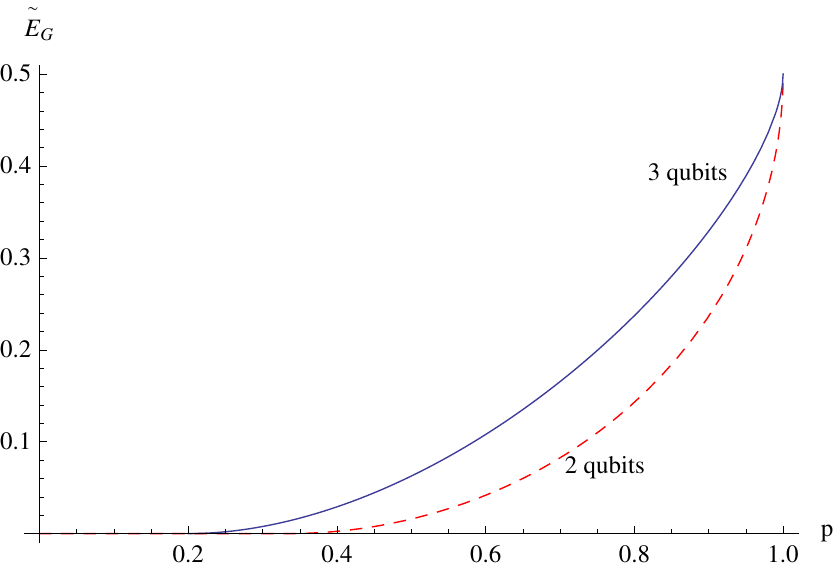}
\par\end{centering}

\caption{\label{fig:isotropic} Approximation of the geometric measure of entanglement
$\tilde{E}_{G}$ for isotropic states of three-qubits given in (\ref{eq:iso})
as a function of $p$ (solid line) compared to the two-qubit case
(dashed line).}

\end{figure}
In this section we apply our algorithm to 3-qubit states with unknown
value of $E_{G}$. If $d$ is the dimension of the total Hilbert space,
then for any $\rho$ there always exists an optimal decomposition
with at most $d^{2}$ elements \cite{Streltsov2010}. A decomposition
$\left\{ p_{i},\ket{\psi_{i}}\right\} $ of a state $\rho=\sum_{i}p_{i}\ket{\psi_{i}}\bra{\psi_{i}}$
is called optimal if its average entanglement is equal to the geometric
measure of entanglement: $\sum_{i}p_{i}E_{G}\left(\ket{\psi_{i}}\right)=E_{G}\left(\rho\right)$.
In order to make sure that the algorithm always has the chance to
find the optimal decomposition, all minimizations in this section
were done over decompositions into $d^{2}=2^{6}=64$ pure states.
In order to do the computation within a reasonable time we used the
parameter $\varepsilon=10^{-7}$.

\subsubsection{Isotropic states}

Isotropic states of three qubits have the form \begin{equation}
\rho=p\ket{\mathrm{GHZ}}\bra{\mathrm{GHZ}}+\frac{1-p}{8}\openone,\label{eq:iso}\end{equation}
with $\ket{\mathrm{GHZ}}=\frac{1}{\sqrt{2}}\left(\ket{000}+\ket{111}\right)$.
They are known to be fully separable if and only if $p\leq\frac{1}{5}$
\cite{Dur1999}. We apply our algorithm to these states with parameter
$\varepsilon=10^{-7}$ for $p>\frac{1}{5}$. The result is shown in
Figure \ref{fig:isotropic} (solid line). The plot can be compared
to the geometric measure of entanglement of the isotropic states of
two qubits, see dashed line in Figure \ref{fig:isotropic}. In the
limit $p\rightarrow1$ the state becomes the pure GHZ state with  $E_{G}\left(\ket{\mathrm{GHZ}}\right)=\frac{1}{2}$
\cite{Wei2003}.

\subsubsection{XX model}

As a final example we apply our algorithm to the isotropic XX model
of 3 qubits in a constant magnetic field. The corresponding Hamiltonian
is given by \cite{Lieb1961,Katsura1962}\begin{equation}
H=\frac{B}{2}\sum_{i=1}^{3}\sigma_{i}^{z}+J\sum_{i=1}^{3}\left(\sigma_{i}^{x}\sigma_{i+1}^{x}+\sigma_{i}^{y}\sigma_{i+1}^{y}\right)\label{eq:H}\end{equation}
with periodic boundary conditions: $\sigma_{4}^{x}=\sigma_{1}^{x}$
and $\sigma_{4}^{y}=\sigma_{1}^{y}$. In thermal equilibrium the system
is found in the mixed state $\rho=\frac{e^{-\frac{H}{kT}}}{Z}$ with
$Z=\mathrm{Tr}\left[e^{-\frac{H}{kT}}\right]$. In the following we
set $k=1$. 

\begin{figure}
\noindent \begin{centering}
\includegraphics[width=1\columnwidth]{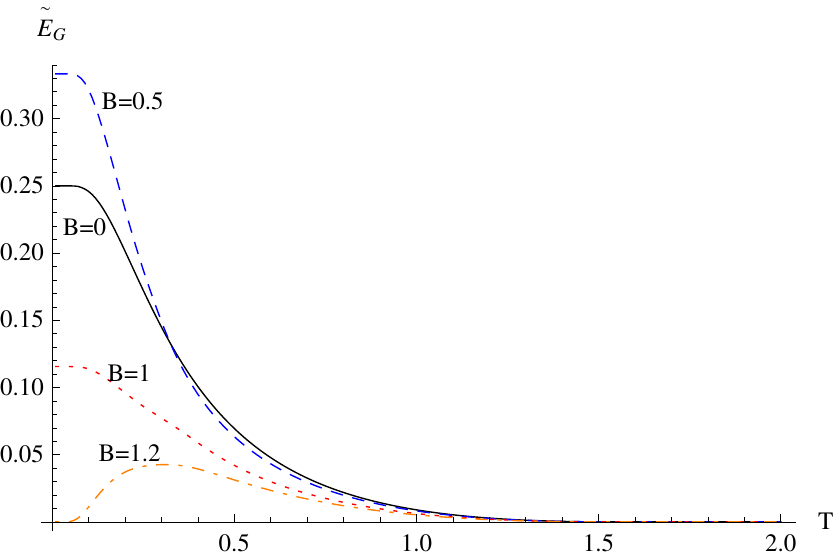}
\par\end{centering}

\caption{\label{fig:XY} Approximation of the geometric measure of entanglement
$\tilde{E}_{G}$ plotted as function of the temperature $T$ for $\rho=\frac{e^{-\frac{H}{kT}}}{Z}$
with $H$ given in (\ref{eq:H}). The parameter $J$ is set to $\frac{1}{2}$,
and $k=1$.}

\end{figure}
The results of the approximation with parameter $\varepsilon=10^{-7}$
are shown in Figure \ref{fig:XY}. They can be compared to the results
for two qubits in \cite[Figure 4]{Wang2001}. For different values
of the magnetic field $B$ we observe a different behavior of the
system in the low temperature limit. This behavior will be explained
in the following. 

Note that the Hamiltonian (\ref{eq:H}) has four nondegenerate eigenvalues
$\pm\frac{3}{2}B$, and $4J\pm\frac{1}{2}B$. Further the following
two eigenvalues are degenerated twice: $-2J\pm\frac{1}{2}B$. For
vanishing magnetic field the ground state of the system is a mixture
of the four eigenstates corresponding to the eigenvalue $-2J$ with
equal probabilities. In this case we get $\tilde{E}_{G}\approx\frac{1}{4}$
for $T\rightarrow0$, see solid curve in Figure \ref{fig:XY}. For
small nonzero magnetic field $0<B<2J$ the ground state of the system
is the mixture of the eigenstates corresponding to the eigenvalue
$-2J-\frac{1}{2}B$. As can be seen from the dashed curve in Figure
\ref{fig:XY}, for $T\rightarrow0$ the approximation becomes $ $$\tilde{E}_{G}\approx\frac{1}{3}$
in this case. In the case $B=2J$, there are three eigenstates corresponding
to the smallest eigenvalue $-3J$. The approximated value for $T\rightarrow0$
in this case becomes $\tilde{E}_{G}\approx0.116$, see dotted curve
in Figure \ref{fig:XY}. Finally, for $B>2J$ the ground state is
the product state $\ket{111}$, and the entanglement vanishes for
$T\rightarrow0$, as is seen from the dot-dashed curve in Figure \ref{fig:XY}. 

\begin{figure}
\noindent \begin{centering}
\includegraphics[width=1\columnwidth]{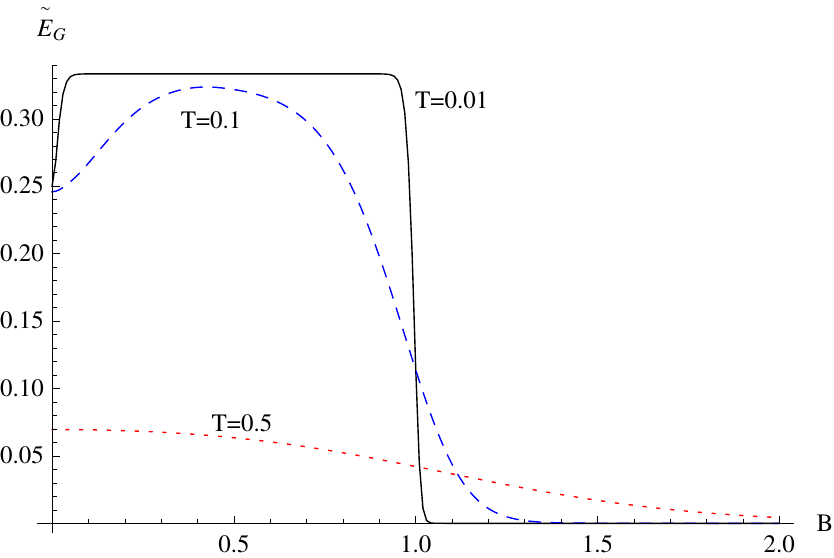}
\par\end{centering}

\caption{\label{fig:XY-1} Approximation of the geometric measure of entanglement
$\tilde{E}_{G}$ for fixed values of $T$ plotted as function of the
magnetic field $B$. The parameter $J$ is set to $\frac{1}{2}$,
and $k=1$.}

\end{figure}
In Figure \ref{fig:XY-1} we show the plot of $\tilde{E}_{G}$ as
function of the magnetic field $B$ for three different temperatures
$T$. For $T\rightarrow0$ we observe that $\tilde{E}_{G}$ becomes
a nonanalytic function of $B$ for two different values of the magnetic
field, namely for $B=0$ and $B=2J$. This is a significant difference
to the two-qubit case, where such behavior occurred only for a single
value of $B$ \cite[Figure 5]{Wang2001}.

\section{\label{sec:Concluding-remarks}Concluding remarks}

In this paper we presented an algorithm for approximating the geometric
measure of entanglement for arbitrary multipartite mixed states. The
algorithm is based on a connection between the geometric measure of
entanglement and the fidelity \cite{Streltsov2010}. It is easy implementable,
since it implies only the solution of an eigenproblem and finding
a singular value decomposition. We tested our algorithm on bipartite
and multipartite mixed states, where an exact formula for the geometric
measure of entanglement is known. In all cases we found convergence
to the exact value. For two qubits, the performance of our algorithm
is comparable to the performance of the algorithms based on conjugate
gradients. We also applied our algorithms to the isotropic state of
three qubits, and the three-qubit XX-model with external magnetic
field. 

In our tests on bipartite mixed states with known value of the geometric
measure of entanglement our algorithm always converged to the correct
value within a given precision. It remains an open question whether
this is always the case. For quantum states with more than two parties
the algorithm can converge to wrong values with nonzero probability.
In general more than one run of the algorithm with different initial
parameters should be performed.

\bibliographystyle{apsrev4-1}
\bibliography{literature}

\end{document}